\documentclass[11pt]{article}
\usepackage{amsthm,amsmath,amssymb}
\usepackage{subcaption}
\usepackage[usenames,dvipsnames]{color}
\usepackage[pdftex,breaklinks,colorlinks,
citecolor={BlueViolet}, linkcolor={Blue},urlcolor=Maroon]{hyperref}
\usepackage{tikz,pgfkeys}
\usetikzlibrary{arrows.meta}
\usetikzlibrary{quotes}
\usetikzlibrary{decorations.pathreplacing}
\usepackage{graphicx}
\usepackage{charter,eulervm}%
\usepackage{multirow,booktabs,array}
\usepackage{tabularx,enumerate}
\usepackage[final,expansion=alltext,protrusion=true]{microtype}
\usepackage[left]{lineno} 

\theoremstyle{plain} 
\newtheorem{theorem}{Theorem}[section]
\newtheorem{lemma}[theorem]{Lemma}
\newtheorem{corollary}[theorem]{Corollary}

\newtheorem{observation}[theorem]{Observation}

\renewcommand{\gg}{G_t}
\renewcommand{\ss}{R}   

\usepackage{graphicx}

\title{Cycle Extendability of Hamiltonian Strongly Chordal Graphs}
\author{Guozhen Rong\thanks{School of Computer Science and Engineering, Central South University, Changsha, China. \href{mailto:rongguozhen@csu.edu.cn} {\tt rongguozhen@csu.edu.cn, jxwang@csu.edu.cn}}
	\and Wenjun Li\thanks{Hunan Provincial Key Laboratory of Intelligent Processing of Big Data on Transportation, Changsha University of Science and Technology, Changsha, China. \href{mailto:lwjcsust@csust.edu.cn} {\tt lwjcsust@csust.edu.cn}}
	\and
	Jianxin Wang\footnotemark[1]
	\and
	Yongjie Yang\thanks{Faculty of Human and Business Sciences, Saarland University, Saarbr\"{u}cken, Germany. \href{yyongjiecs@gmail.com} {\tt yyongjiecs@gmail.com}}
}

\begin{document}
\maketitle

\begin{abstract}
 	In 1990, Hendry conjectured that all Hamiltonian chordal graphs are cycle extendable. After a series of papers confirming the conjecture for a number of graph classes, the conjecture is yet refuted by Lafond and Seamone in 2015. Given that their counterexamples are not strongly chordal graphs and they are all only $2$-connected, Lafond and Seamone asked the following two questions: (1) Are Hamiltonian strongly chordal graphs cycle extendable? (2) Is there an integer~$k$ such that all $k$-connected Hamiltonian chordal graphs are cycle extendable?
		Later, a conjecture stronger than Hendry's is proposed. In this paper, we resolve all these questions in the negative.
		On the positive side, we add to the list of cycle extendable graphs two more graph classes, namely, Hamiltonian $4$-\textsc{fan}-free chordal graphs where every induced $K_5 - e$ has true twins, and Hamiltonian  $\{4\textsc{-fan}, \overline{A} \}$-free chordal graphs.
\end{abstract}


\section{Introduction}
	\label{sec-introduction}
	A graph is \emph{Hamiltonian} if it has a {\emph{Hamiltonian cycle}}.
	Investigating sufficient conditions for the existence of a Hamiltonian cycle has been a prevalent topic, initiated by the seminal work of Dirac~\cite{Dirac1952}.
	A commonly used scheme to show the existence of a Hamiltonian cycle is to derive a contradiction to the assumption that the graph has no Hamiltonian cycle, by means of extending an assumed longest non-Hamiltonian cycle to a longer cycle. After observing this, Hendry~\cite{Hendry89, Hendry90, Hendry91} proposed the concept of cycle extendability.
	Concretely,	a cycle~$C$ is \emph{extendable} if there exists a cycle~$C'$ which contains all vertices of~$C$ plus one more vertex not in~$C$.
	A graph is {\emph{cycle extendable} if all non-Hamiltonian cycles of the graph are extendable. The notion of cycle extendability is related to the well-studied notion of pancyclicity.
		Recall that a graph on~$n$ vertices is \emph{pancyclic} if it contains a cycle of length~$\ell$ for every~$\ell$ such that $3 \le \ell \le n$.
		Clearly, every cycle extendable graph containing at least one cycle of length~$3$ is pancyclic, and every pancyclic graph is Hamiltonian.
		In~\cite{Hendry90}, Hendry studied several sufficient conditions for a graph to be cycle extendable, and in the conclusion he put forward a ``reverse'' notion of cycle extendability, namely the cycle reducibility.
		Precisely, a graph is {\emph{cycle reducible}} if for every cycle~$C$ containing more than 3 vertices in the graph there exists a cycle~$C'$ which consists of $|V(C)|-1$ vertices of~$C$.
		In light of the facts that (1) a graph is cycle reducible if and only if it is a chordal graph, and (2) Hamiltonian chordal graphs are pancyclic, Hendry~\cite{Hendry90} proposed the following conjecture.
\smallskip

{\bf{Hendry's Conjecture.}} Hamiltonian chordal graphs are cycle extendable.
\smallskip	

		Recall that chordal graphs are the graphs without holes (induced cycles of length at least four) as induced subgraphs.
		Since the work of Hendry~\cite{Hendry90}, the above conjecture has  received considerable attention.
		Remarkably, Hendry's conjecture has been confirmed for a number of special graph classes including planar Hamiltonian chordal graphs~\cite{Jiang02}, Hamiltonian interval graphs~\cite{Abueida06S,ChenFGJ06}, Hamiltonian split graphs~\cite{Abueida06S}, etc. In 2013,  Abueida, Busch, and Sritharan~\cite{AbueidaBS13} added to the list of cycle extendable graph classes the Hamiltonian spider intersection graphs, a superclass of Hamiltonian interval and Hamiltonian split graphs.

		Though all these confirmative works continuously fill the gap step by step and provide more and more evidence towards a positive answer to Hendry's conjecture, the conjecture is, somewhat surprisingly, eventually disproved the first time by Lafond and Seamone~\cite{LafondS15} in~2015. Particularly, Lafond and Seamone derived a counterexample with only~$15$ vertices. Based on this counterexample, they also showed that for any $n \ge 15$, a counterexample on~$n$ vertices exists.
		Nevertheless, Lafond and Seamone's work is not the end of the story, as there are many interesting subclasses of chordal graphs for which whether Hendry's conjecture holds still remained open.
		Notably, as all counterexamples constructed in~\cite{LafondS15} contain an induced \emph{$3${\textsc{-sun}}} which is a forbidden induced subgraph of strongly chordal graphs (see Figure~\ref{fig:suns} for a $3$-{\textsc{sun}}), and contain at least one degree-$2$ vertex, Lafond and Seamone proposed two questions as follows. (For notions in the following discussions, we refer to the next section for the formal definitions.)
		
\smallskip

			{\bf{Question~1.}} {{Are Hamiltonian strongly chordal graphs cycle extendable?}
\smallskip
		
			{\bf{Question~2.}} {{Is there an integer $k > 2$ such that all $k$-connected Hamiltonian chordal graphs are cycle extendable?}}
\smallskip

		Later, based on the concept of $\ss$-cycle extendability first studied in~\cite{BeasleyB2005},  a more general conjecture was proposed.
		Let~$\ss$ be a nonempty subset of positive integers. A cycle~$C$ in a graph~$G$ is $\ss$-extendable if there exists another cycle~$C'$ in~$G$ which consists of all vertices of~$C$ and~$i$ additional vertices for some integer $i\in \ss$. A graph is $\ss$-cycle extendable if every non-Hamiltonian cycle of the graph is $\ss$-extendable. Clearly, $\{1\}$-cycle extendable graphs are exactly cycle extendable graphs. After observing that the counterexamples of Lafond and Seamone are $\{1, 2\}$-cycle extendable, Arangno~\cite{Arangnothesis} put forward in his Ph.D.\ thesis the following conjecture.
\smallskip	
	
			{\bf{Arangno's Conjecture.}} Hamiltonian chordal graphs are $\{1, 2\}$-cycle extendable.
\smallskip

		Our first contribution is the following theorem, which directly refutes Arangno's conjecture and provides negative answers to Questions~1--2~\footnote{Independent of our work, Lafond et al.~\cite{LafondSS20} obtained similar results for these conjectures.}.
		
		\begin{theorem}\label{thm:counter-k-connected}
			Let~$\ss$ be a nonempty set of positive integers, and let~$t$ be the maximum integer in~$\ss$. Then, for all integers $k \ge 0$ and $n \ge 14 + t + 2k$, there exists a $(2 + k)$-connected Hamiltonian strongly chordal graph with~$n$ vertices that is not $\ss$-cycle extendable.
		\end{theorem}

On the positive side, we add to the list of graph classes fulfilling Hendry's conjecture two subclasses of Hamiltonian chordal graphs. Recently, Gerek~\cite{Gerekthesis} proved that Hendry's conjecture holds for Hamiltonian Ptolemaic graphs. Recall that Ptolemaic graphs are exactly chordal graphs that are $3$-{\textsc{fan}}-free~\cite{Howorka}. Though that we could not extend this result to Hamiltonian $4$-{\textsc{fan}}-free chordal graphs, we show that Hendry's conjecture holds for two subclasses of $4$-{\textsc{fan}}-free chordal graphs, namely, Hamiltonian $4$-{\textsc{fan}}-free chordal graphs where every induced $K_5-e$ has true twins, and Hamiltonian $\{4\textsc{-fan}, \overline{A}\}$-free chordal graphs. See Figure~\ref{fig:A-bar} for~$3\textsc{-fan}$, $4\textsc{-fan}$, $K_5-e$, and~$\overline{A}$.

\begin{figure}[h]
\centering
{\includegraphics[width=0.95\textwidth]{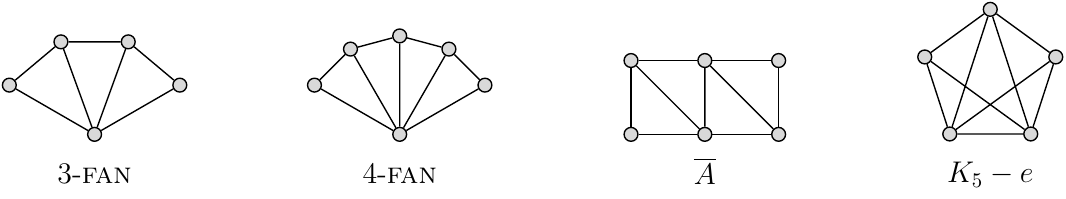}}
\caption{$3\textsc{-fan}$, $4\textsc{-fan}$, $\overline{A}$, and $K_5-e$}
\label{fig:A-bar}
\end{figure}
		
		\begin{theorem}\label{thm:fan-k5}
			Hamiltonian $4\textsc{-fan}$-free chordal graphs where every induced $K_5 - e$ (if there are any)  has true twins are cycle extendable.
		\end{theorem}

		\begin{theorem}\label{thm:fan-free-power}
			Hamiltonian $\{4\textsc{-fan}, \overline{A}\}$-free chordal graphs are cycle extendable.
		\end{theorem}

In addition, it is known that the class of $k$-leaf powers for all integers~$k\geq 1$ are a natural subclass of strongly chordal graphs~\cite{NevriesR16}.
Gerek's result also implies that Hendry's conjecture holds for Hamiltonian $k$-leaf powers for $k=1, 2, 3$ because they are subclasses of Ptolemaic graphs. 
Because $4$-leaf powers are free of induced $4$-{\textsc{fan}}s and contain induced $K_5-e$ only with true twins~\cite{Rautenbach06}, we obtain the following corollary as a consequence of Theorem~\ref{thm:fan-k5}.

\begin{corollary}\label{cor:4-leaf-power}
Hamiltonian $4$-leaf powers are cycle extendable.
\end{corollary}

We also would like to mention, in passing, that $\{4\textsc{-fan}, \overline{A} \}$-free chordal graphs contain the well-partitioned chordal graphs coined by Ahn~et~al.~\cite{AhnWG2020} very recently.
		
\bigskip
		
		\noindent{\it{Organization.}} In Section~\ref{sec-pre}, we provide basic notions used in our paper. Section~\ref{sec:Counterexamples} is devoted to the proof of Theorem~\ref{thm:counter-k-connected}, and Section~\ref{sec-positive} composites the proofs of Theorems~\ref{thm:fan-k5}--\ref{thm:fan-free-power}. We conclude our study in Section~\ref{sec-conclusion}.

		\section{Preliminaries}
		\label{sec-pre}
		We assume the reader is familiar with the basics of graph theory. We reiterate numerous important notions used in our discussions, and refer to~\cite{Douglas2000} for notions used in the paper but whose definitions are not provided in this section. By convention, for an integer~$i$, we use~$[i]$ to denote the set of all positive integers at most~$i$.
		
		All graphs considered in this paper are finite, undirected, and simple. The vertex set and the edge set of a graph~$G$ are denoted by~$V(G)$ and~$E(G)$, respectively.
		We use~$uv$ to denote an edge between two vertices~$u$ and~$v$. For a vertex $v \in V(G)$, $N(v) = \{ u \mid uv \in E(G) \}$ denotes the \emph{(open) neighborhood} of~$v$, and $N[v] = N(v) \cup \{v\}$ denotes the \emph{closed neighborhood} of~$v$. The \emph{degree} of~$v$ is the cardinality of~$N(v)$. For a subset $X \subseteq V(G)$, let $N(X) = \bigcup_{v \in X} N(v) \setminus X$ and $N[X] = N(X) \cup X$. The subgraph induced by~$X$ is denoted by~$G[X]$. For brevity, we write $G - X$ for the subgraph of~$G$ induced by $V(G)\setminus X$. When $X = \{x\}$, we simply use the shorthand $G- x$ for $G - X$.
		
		We say that two vertices are {\emph{true twins}} if they have the same closed neighborhood. A vertex~$v$ is a true-twin vertex (true-twin for short) if there is another vertex~$u$ such that~$v$ and~$u$ are true twins. In this case, we call~$v$ a true-twin for~$u$. A {\emph{true-twins pair}} refers to a pair $\{u, v\}$ such that~$u$ and~$v$ are true twins.
		
		A vertex is \emph{universal} in a graph if it is adjacent to all other vertices. A vertex of degree~$0$ is called an \emph{isolated} vertex.
		A vertex~$v$ in a graph~$G$ is \emph{simplicial} if~$N[v]$ is a clique.
		Moreover, if the vertices in~$N(v)$ can be ordered as $(v_1, v_2, \ldots, v_k)$ such that $N[v_1] \subseteq N[v_2] \subseteq \cdots \subseteq N[v_k]$, where $k=|N(v)|$, then we say that~$v$ is \emph{simple}. A simple vertex is always simplicial.

\begin{figure}[h]
\centering
{
\includegraphics[width=0.75\textwidth]{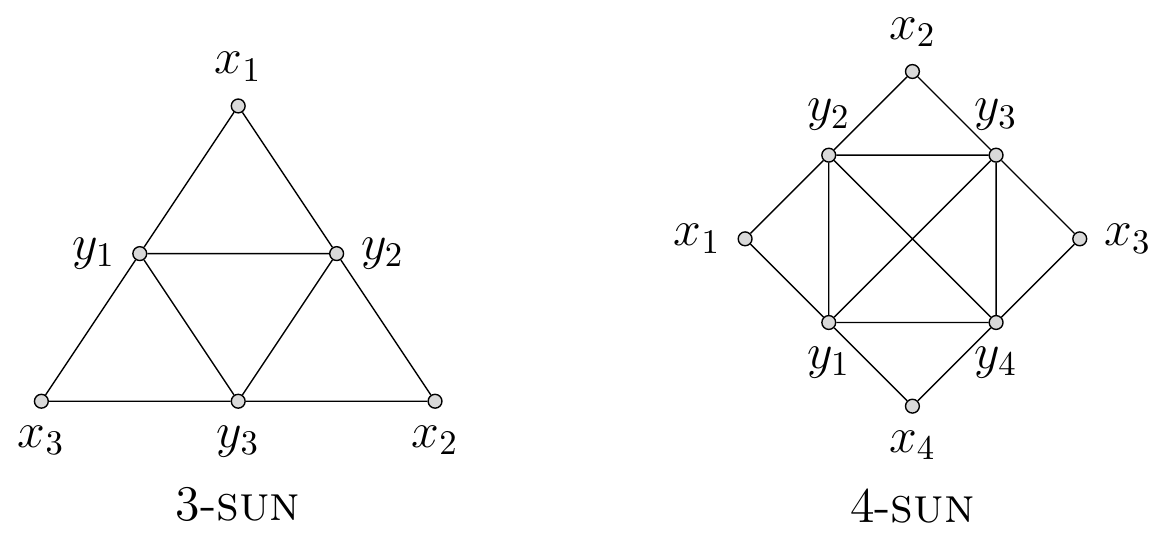}
}
\caption{$3$-\textsc{sun} and $4$-\textsc{sun}}
\label{fig:suns}
\end{figure}
		A \emph{hole} in a graph~$G$ is an induced cycle of~$G$ of length at least~$4$. An {\emph{independent set}}  of~$G$ is a subset~$S$ of vertices such that~$G[S]$ contains only isolated vertices. A {\emph{clique}} is a subset~$S$ of vertices such that there is an edge between every pair of vertices in~$S$.
		For an integer $k \ge 3$, a \emph{$k$-\textsc{sun}} is a graph of~$2k$ vertices which can be partitioned into an independent set $X = \{ x_1, \ldots, x_k \}$ and a clique $Y = \{ y_1, \ldots, y_k \}$  such that for every $i \in [k]$, $x_i$ is only adjacent to $y_{i}$ and $y_{(i \pmod k)+1}$. See Figure~\ref{fig:suns} for the $3$-\textsc{sun}} and $4$-\textsc{sun}}.
		
		A graph is \emph{chordal} if it does not contain any holes.
		\emph{Strongly chordal graphs} are chordal graphs without induced $k$-{\textsc{sun}}s for all $k \ge 3$.
		Hence, the minimal forbidden induced subgraphs of strongly chordal graphs are $k$-{\textsc{sun}}s and holes, none of which contains a universal vertex or a true-twins pair.
		Strongly chordal graphs admit an ordering characterization.
		In particular, a \emph{simple elimination ordering} of a graph~$G$ is an ordering $(v_1, v_2, \ldots, v_n)$ over~$V(G)$ such that for every $i\in [n]$,~$v_i$ is simple in the subgraph of~$G$ induced by $\{v_i, v_{i+1}, \dots, v_n\}$. It has been proved that a graph is a strongly chordal graph if and only if it admits a simple elimination ordering~\cite{Farber83}.

		A graph~$G$ is a $k$-leaf power if there exists a tree~$T$ such that (1) the vertices of~$G$ one-to-one correspond to the leaves of~$T$, and (2) for every two vertices $u, v\in V(G)$, it holds that $uv \in E(G)$ if and only if the distance between~$u$ and~$v$ in~$T$ is at most~$k$. It is a folklore that $k$-leaf powers are strongly chordal graphs~\cite{NevriesR16}.
		
A {\emph{path}} between two vertices~$v$ and~$u$ is a sequence of distinct vertices such that~$v$ and~$u$ are the first and last vertices in the sequence and, moreover, every two consecutive vertices are adjacent. Such a path is called a $v$-$u$ path. For a $v$-$u$ path~$P$, and two vertices $v'\in N(v)$ and $u'\in N(u)$ not in the path, $v' P u'$ denotes the $v'$-$u'$ path obtained from~$P$ by putting~$v'$ before~$v$ and putting~$u'$ after~$u$. The length of a path refers to the number of vertices in the path minus one.

		A graph is {\emph{connected}} if it has only one vertex or between every two vertices there exists a path in the graph. A graph is {\emph{$k$-connected}} if it is connected after the deletion of any subset of at most $k-1$ vertices.

\section{Negative Results}\label{sec:Counterexamples}
This section is devoted to the proof of Theorem~\ref{thm:counter-k-connected}.
Our counterexamples are based on the graphs~$\widehat{H}$ and~$\widehat{H}^-$ shown in Figure~\ref{fig:base-graphs}, where the graph~$\widehat{H}$ is obtained from the $15$-vertex counterexample of Lafond and Seamone by adding one edge $a u_3$~\cite{LafondS15}. The following two observations follow immediately from the definitions of strongly chordal graphs and Hamiltonian graphs respectively.
		
		\begin{observation}\label{prop:universal-strongly}
			A strongly chordal graph remains strongly chordal after adding or deleting a universal vertex or a true-twin for any vertex.
		\end{observation}
		
		\begin{observation}\label{prop:Hamiltonian}
			A Hamiltonian graph remains Hamiltonian after adding a universal vertex or a true-twin vertex for any vertex.
		\end{observation}

\begin{figure}[h]
\centering
{
\includegraphics[width=\textwidth]{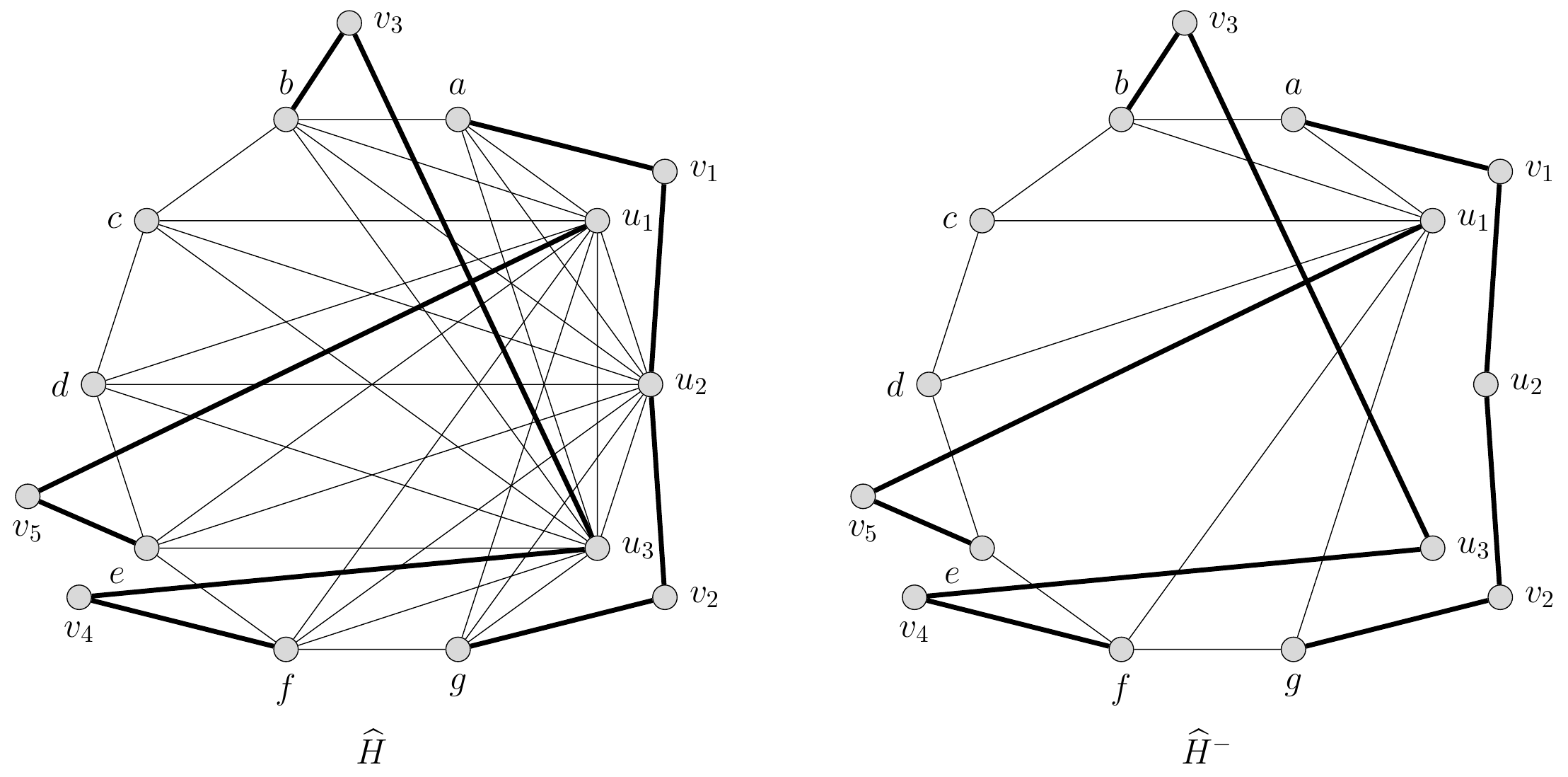}
\caption{Two graphs~$\widehat{H}$ and~$\widehat{H}^-$. 
Heavy edges are in bold.
				The graph $\widehat{H} - \{v_1,v_2,v_3,v_4,v_5\}$ is composed by a $7$-vertex path $abcdefg$ plus three universal vertices $u_1$,~$u_2$, and~$u_3$, and $\widehat{H}^-$ is obtained from $\widehat{H}$ by the deletion of the edge $u_1 e$ and the nonheavy edges incident to~$u_2$ and~$u_3$.
			}
			\label{fig:base-graphs}
}
\end{figure}

		\begin{lemma}\label{lem:TwoSCGraphs}
			$\widehat{H}$ is a Hamiltonian strongly chordal graph.
		\end{lemma}
		
		\begin{proof}
			The graph~$\widehat{H}$ is Hamiltonian since $a v_1 u_2 v_2 g f v_4 u_3 v_3 b c d e v_5 u_1 a$ is a Hamiltonian cycle of~$\widehat{H}$.
			It remains to prove that~$\widehat{H}$ is a strongly chordal graph.
			It is fairly easy to check that the vertices $v_1$,~$v_2$,~$v_3$,~$v_4$, and~$v_5$ are simple vertices in~$\widehat{H}$. Let  $H'=\widehat{H} - \{v_1, v_2, v_3, v_4, v_5\}$. Since simple vertices are not in any induced holes and $k$-{\textsc{sun}}s for all $k\geq 3$, it suffices to prove that~$H'$ is a strongly chordal graph. To this end, observe that~$H'$ consists of a $7$-vertex path $abcdefg$ and three universal vertices~$u_1$,~$u_2$, and~$u_3$. A path is clearly a strongly chordal graph. Then, by Observation~\ref{prop:universal-strongly},~$H'$ is a strongly chordal graph too.
		\end{proof}
		
		The vertices~$v_1$,~$v_2$,~$v_3$,~$v_4$, and~$v_5$ are all of degree-$2$ in $\widehat{H}$. We call the edges incident to them \emph{heavy edges} (see Figure~\ref{fig:base-graphs}). 
		\begin{observation}\label{obs:heavy-edges-cycle}
			Every cycle in $\widehat{H}$ containing $\{v_1,v_2,v_3,v_4,v_5\}$ must contain all the heavy edges.
		\end{observation}
		
		\begin{lemma}\label{lem:heavy-edges-equal}
			$C$ is a cycle containing all heavy edges in~$\widehat{H}$ if and only if~$C$ is a cycle containing all heavy edges in $\widehat{H}^-$.
		\end{lemma}

		\begin{proof}
			Note that the heavy edges of~$\widehat{H}$ and~$\widehat{H}^-$ are the same. The ``if'' direction is trivial since $\widehat{H}^-$ is a subgraph of~$\widehat{H}$.
			For the ``only if'' direction, assume that~$C$ is a cycle containing all heavy edges in~$\widehat{H}$. As $ev_5, u_1v_5 \in E(C)$, we have that $u_1e \notin E(C)$. As $v_1u_2, v_2u_2 \in E(C)$, the other edges incident to~$u_2$, which are all nonheavy edges, are not in~$C$. For the same reason, all nonheavy edges incident to~$u_3$ are not in~$C$. Recall that~$\widehat{H}^-$ is obtained from~$\widehat{H}$ by the deletion of~$u_1e$ and the nonheavy edges incident to~$u_2$ and~$u_3$. Therefore, $E(C) \subseteq E(\widehat{H}^-)$, which means that~$C$ is a cycle of~$\widehat{H}^-$.
		\end{proof}
		
		Now we study a nonextendable cycle in $\widehat{H}$.
		
		\begin{lemma}\label{lemma:exceptional-cycle}
			The cycle $C = a v_1 u_2 v_2 g u_1 v_5 e f v_4 u_3 v_3 b a$ in~$\widehat{H}$ is not extendable.
		\end{lemma}
		
		\begin{proof}
			Note that the cycle~$C$ contains all vertices of~$\widehat{H}$ except the two vertices~$c$ and~$d$.
			Suppose for contradiction that~$C$ admits an extension~$C'$. Clearly,~$C'$ contains $\{v_1,v_2,v_3,v_4,v_5\}$, and it holds that~$C'$ is a Hamiltonian cycle of either $\widehat{H} - c$ or $\widehat{H} - d$. By Observation~\ref{obs:heavy-edges-cycle} and Lemma~\ref{lem:heavy-edges-equal},~$C'$ is either a Hamiltonian cycle of $\widehat{H}^- - c$ or a Hamiltonian cycle of $\widehat{H}^- - d$ and, moreover,~$C'$ contains all heavy edges.

			If~$C'$ is a Hamiltonian cycle of $\widehat{H}^- - c$, then as the vertex~$d$ has degree~$2$ in  $\widehat{H}^- - c$, the two edges~$du_1$ and~$de$ incident to~$d$ are contained in the cycle~$C'$. However, these two edges together with the heavy edges incident to~$v_5$ form a cycle $d u_1 v_5 e d$ of length four, contradicting that~$C'$ is a Hamiltonian cycle of~$\widehat{H}^- - c$.
			
			Finally, if~$C'$ is a Hamiltonian cycle of~$\widehat{H}^- - d$, then as~$c$ and~$e$ are both degree-$2$ vertices in~$\widehat{H}^- - d$,~$C'$ must contain the  edges~$cb$,~$u_1c$, and~$fe$. However, these edges together with the heavy edges~$bv_3$,~$v_3u_3$,~$u_3v_4$,~$v_4f$,~$ev_5$, and~$v_5u_1$ form a cycle $b v_3 u_3 v_4 f e v_5 u_1 c b$ of length nine, contradicting that~$C'$ is a Hamiltonian cycle of $\widehat{H}^- - d$.
		\end{proof}
		
Now we are ready to give the proof of Theorem~\ref{thm:counter-k-connected}. We first show the proof for the special case where $k=0$. The proof for the case where $k\geq 1$ is built upon the proof for $k=0$ by adding a number of vertices, and is given subsequently.

\bigskip

\begin{proof}[Proof of Theorem~\ref{thm:counter-k-connected} for $k=0$]
Let~$\ss$ be a nonempty set of positive integers, and let~$t$ be the maximum integer in~$\ss$. In the following, we show that for every integer $n\geq 14+t$, there exists a $2$-connected Hamiltonian strongly chordal graph with~$n$ vertices that is not $[t]$-cycle extendable, and hence not $\ss$-cycle extendable. To this end, we modify the graph~$\widehat{H}$ into a graph~$\gg$ as follows: we first replace the edge~$cd$ by a path $P = c z_1 z_2 \cdots z_{t-1} d$ of $t+1$ vertices, and then we add edges so that each~$z_i$, where $i\in [t-1]$, is adjacent to~$u_1$,~$u_2$, and~$u_3$. For convenience, we use~$z_0$ to denote~$c$, use~$z_{t}$ to denote~$d$, and define~$Z = \{z_0, \ldots, z_t\}$. Let $\gg^-$ be the graph obtained from~$\gg$ by deleting all nonheavy edges incident to~$u_2$ and~$u_3$, and deleting the edge~$eu_1$. Clearly, both~$\gg$ and~$\gg^-$ contain exactly $14+t$ vertices.
			See Figure~\ref{fig:t-extend-base-graphs} for illustrations of~$\gg$ and~$\gg^-$.
			
			The graph~$\gg$ is Hamiltonian since $a v_1 u_2 v_2 g f v_4 u_3 v_3 b z_0 \cdots z_t e v_5 u_1 a$ is a Hamiltonian cycle of~$\gg$.
			Analogous to the proof of Lemma~\ref{lem:TwoSCGraphs}, it can be shown that~$\gg$ is a strongly chordal graph. In particular, it is easy to see that every~$v_i$, where $i\in [5]$, is a simple vertex in~$\gg$. Then, it suffices to show that the graph~$\gg$ without the five vertices~$v_1$, $\dots$, $v_5$ is a strongly chordal graph. This is the case as $\gg-\{v_1, \dots, v_5\}$ consists of a path $a b z_0 z_1 \cdots z_t e f g$ (which is strongly chordal) and three universal vertices~$u_1$,~$u_2$, and~$u_3$.

			Observation~\ref{obs:heavy-edges-cycle} and Lemma~\ref{lem:heavy-edges-equal} also apply to~$\gg$. That is, the following conditions are fulfilled by~$\gg$.
			\begin{itemize}
				\item Every cycle in~$\gg$ containing $\{v_1,v_2,v_3,v_4,v_5\}$ contains all heavy edges.
				\item $C$ is a cycle containing all heavy edges in~$\gg$ if and only if~$C$ is a cycle containing all heavy edges in~$\gg^-$.
			\end{itemize}
We show now that the cycle
\[C = a v_1 u_2 v_2 g u_1 v_5 e f v_4 u_3 v_3 b a\]
in~$\gg^-$ is not $\{i\}$-extendable for all $i\in [t]$.
			We prove this by induction.
			For the base case where $i=1$, our proof goes as follows. Assume for the sake of contradiction that~$C$ can be extended to~$C'$ such that $V(C') = V(C) \cup \{z\}$ where $z\in Z$.
			Since~$z_j$ where $j\in [t-1]$ has degree one in $\gg^{-} - (Z \setminus \{z_j\})$, it holds that $z \notin Z \setminus \{z_0, z_t\}$. Note that $\gg^{-} - (Z \setminus \{z_0\})$ is isomorphic to~$\widehat{H}^- - d$, and $\gg^{-} - (Z \setminus \{z_t\})$ is isomorphic to $\widehat{H}^- - c$. Then, by Lemma~\ref{lemma:exceptional-cycle},~$z$ can neither be $z_0$ nor $z_t$. This completes the proof for the base case.

			Now assuming that $i>1$ and~$C$ is not $[i-1]$-extendable, we claim that the cycle~$C$ in~$\gg^-$ is not $\{i\}$-extendable.
			We prove this by contradiction. Assume for the sake of contradiction that~$C^*$ is an $\{i\}$-extension of~$C$ in~$\gg^-$ such that $V(C^*)=V(C) \cup Z'$ for some $Z' \subseteq Z$ with $|Z'| = i$.
			Let $Z^- = Z \setminus Z'$.
			We claim that $\gg^- - (Z \setminus \{z_0, z_t\})$ is not Hamiltonian: In $\gg^- - (Z \setminus \{z_0, z_t\})$,~$z_t$ and~$v_5$ are degree-$2$ vertices; a Hamiltonian cycle of $\gg^- - (Z \setminus \{z_0, z_t\})$ must contain~$z_t$ and~$v_5$, and hence contains the four edges~$ez_t$,~$z_tu_1$,~$u_1v_5$, and~$v_5e$, which, however, form a cycle of length four.
			By the claim, $Z' \neq \{z_0, z_t\}$.
			Note that vertices of~$Z$ are all degree-$3$ vertices in~$\gg^-$, $Z \setminus Z' \neq \emptyset$ (because $|Z|=t+1$ and $i\leq t$), and every vertex of~$Z'$ in~$C^*$ has degree~$2$. Hence, there exists at least one vertex of~$Z'$ in $\gg^- - Z^-$ with degree~$2$, and there exists no vertex of~$Z'$ in $\gg^- - Z^-$ with degrees~$0$ or~$1$. As $Z' \neq \{z_0, z_t\}$, there are two vertices $z, z' \in Z'$ such that~$z$ has degree~$2$ in $\gg^- - Z^-$ with~$z'$ and~$u_1$ being its two neighbors. It is clear that $z'u_1 \in E(\gg^- - Z^-)$. By replacing the edges~$zz'$ and~$zu_1$ with the edge~$z'u_1$ in~$C^*$ we get a cycle on $V(C^*) \setminus \{z\}$, an $\{ i-1 \}$-extension of~$C$, a contradiction.
			
			A Hamiltonian strongly chordal graph of~$n> 14+t$ vertices which is not $\ss$-cycle extendable can be obtained from the graph~$\gg$ by adding $n-14-t$ true-twins of~$v_1$.
\end{proof}

\begin{figure}[h]
\centering
{
\includegraphics[width=\textwidth]{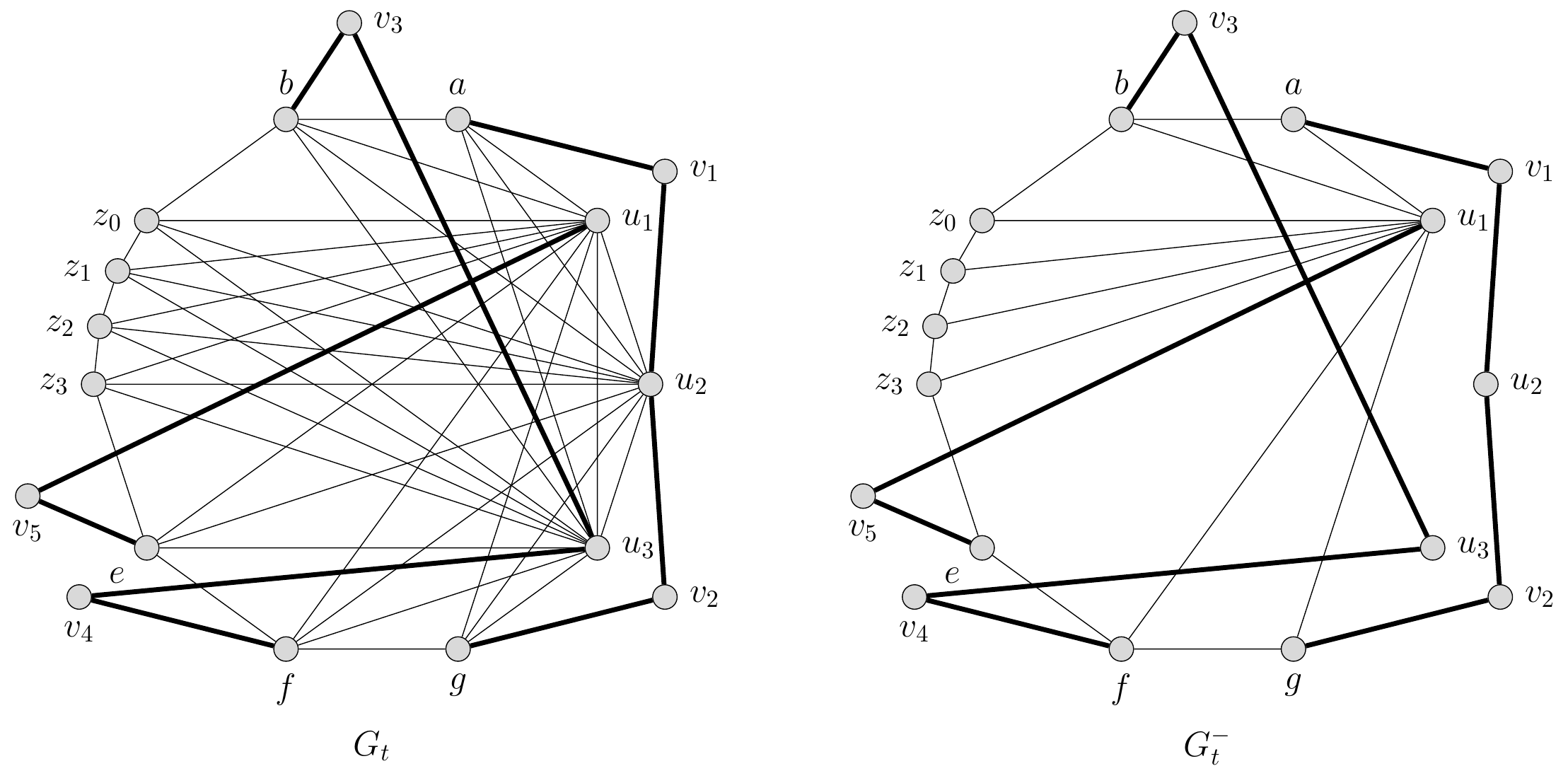}
\caption{Two graphs~$\gg$ and~$\gg^-$ in the proof of Theorem~\ref{thm:counter-k-connected} for $k =0$ and $t=3$.}
\label{fig:t-extend-base-graphs}
}
\end{figure}

		Now, we move on to the proof of Theorem~\ref{thm:counter-k-connected} for the case where $k\geq 1$. We need a few additional notions for our exposition.
		For two positive integers~$p$ and~$q$, a \emph{$(p, q)$-star} is a graph whose vertex set can be partitioned into an independent set~$X$ of~$p$ vertices and a clique~$Y$ of~$q$ vertices such that all vertices in~$X$ are adjacent to all vertices in~$Y$. Moreover, such a partition $(X, Y)$ is called the $(p, q)$-partition of the graph.

\begin{figure}[h]
	\centering
	{
		\includegraphics[width=0.85\textwidth]{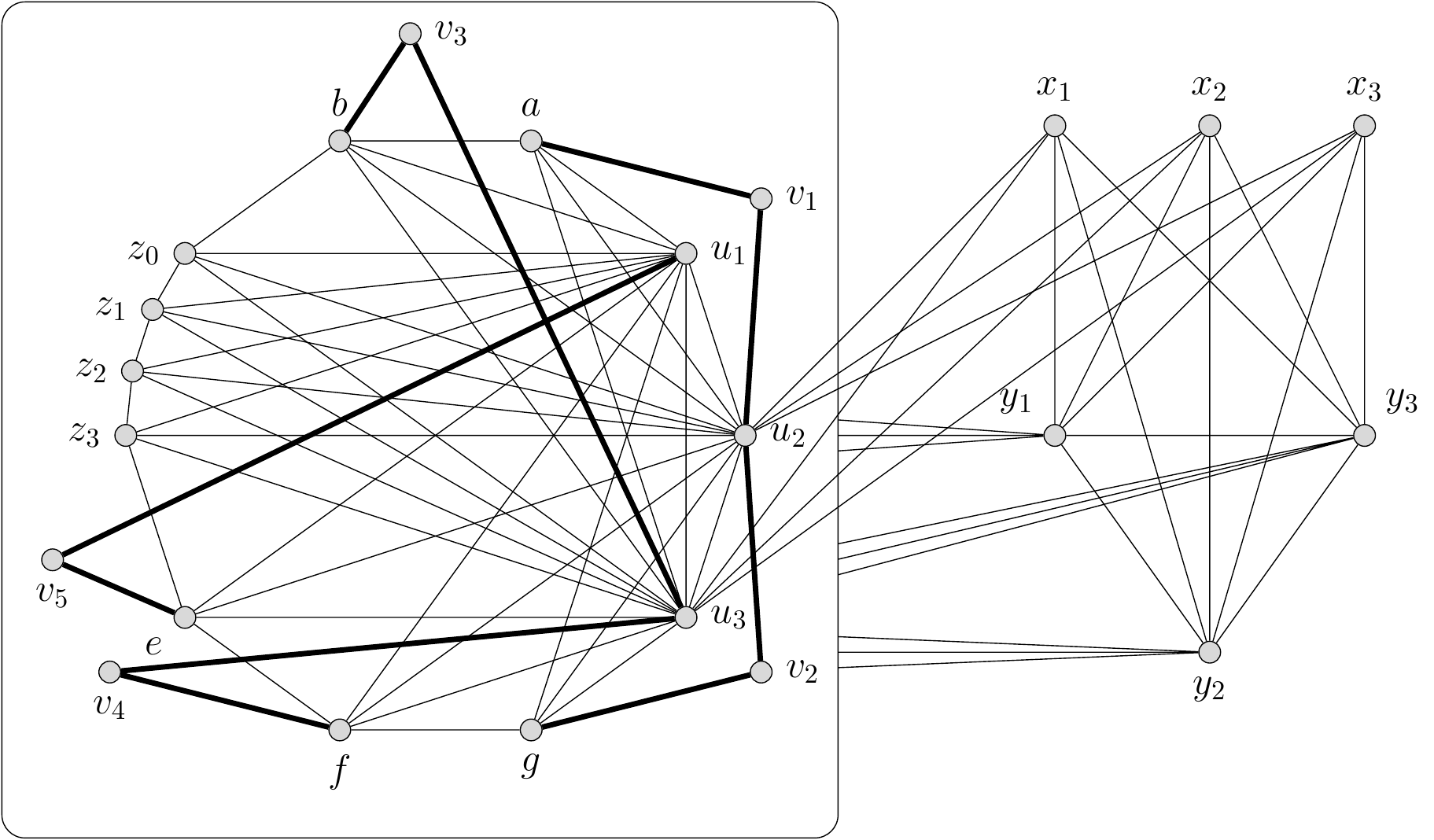}
		\caption{The graph~$\gg^3$. The set of vertices in the left box induces a~$\gg$ for $t=3$, and $\{x_1,x_2,x_3,y_1,y_2,y_3\}$ induces a $(3,3)$-star. Multiedges between a vertex and the left box means that the vertex is adjacent to all vertices in the box.}
		\label{fig:k-connected-a}
	}
\end{figure}

\begin{figure}[h]
	\centering
	{
		\includegraphics[width=0.5\textwidth]{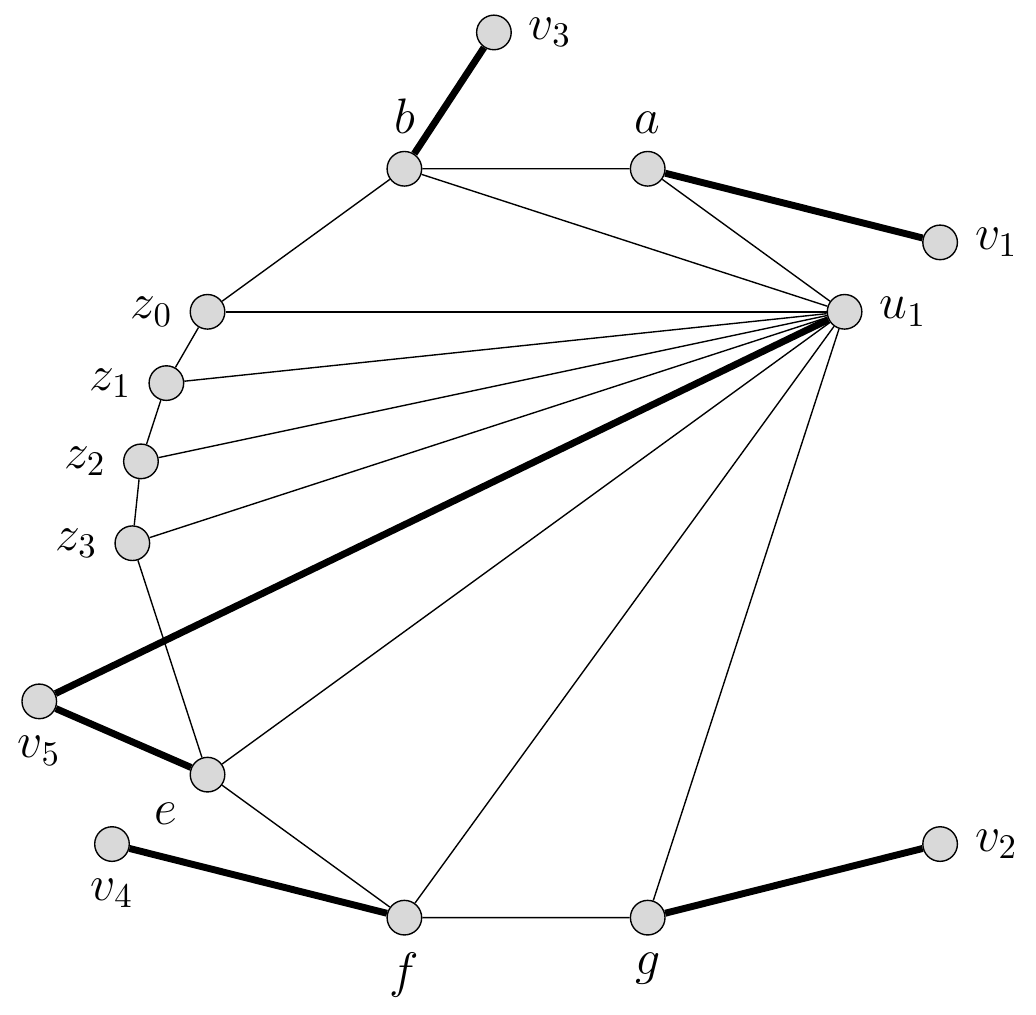}
	}
	\caption{The graph $G_t - \{u_2, u_3\}$ where $t=3$.}
	\label{fig:bbb}
\end{figure}

		\bigskip
		\begin{proof}[Proof of Theorem~\ref{thm:counter-k-connected} for $k\geq 1$]
	Let $\gg$ and $Z = \{z_0, \ldots, z_t\}$ be as defined in the above proof for $k = 0$.
	We construct a graph~$\gg^k$ of $14 + t + 2k$ vertices as follows:
	\begin{enumerate}
		\item take the union of~${\gg}$ and a $(k, k)$-star with the $(k, k)$-partition $(X, Y)$;
		\item add an edge between every vertex in~$Y$ and every vertex in~${\gg}$;
		\item add an edge between~$u_2$ and each vertex of~$X$; and
		\item add an edge between~$u_3$ and each vertex of~$X$.
	\end{enumerate}	
We refer to Figure~\ref{fig:k-connected-a} for an illustration of~$\gg^3$.

We first show that the graph~$\gg^k$ is a $(2+k)$-connected Hamiltonian strongly chordal graph.
\begin{itemize}
\item {\bf{$\gg^k$ is $(2+k)$-connected}}

 As $|Y| = k$ and all vertices in~$Y$ are universal vertices, this is equivalent to showing that $\gg^k-Y$ is $2$-connected, i.e., $\gg^k-Y$ is connected after deleting any arbitrary vertex. This is the case as~${\gg}$ has a Hamiltonian cycle and, moreover, every vertex in~$X$ is adjacent to two vertices in~${\gg}$.

\item {\bf{$\gg^k$ is Hamiltonian}}

Let $x_1 x_2 \cdots x_k$ and  $y_1 y_2 \cdots y_k$ be any arbitrary but fixed orders over~$X$ and~$Y$ respectively. The following Hamiltonian cycle is an evidence that~${\gg^k}$ is Hamiltonian:
	\[ a v_1 u_2 x_1 y_1 x_2 y_2 \cdots x_k y_k v_2 g f v_4 u_3 v_3 b z_0 z_1 \cdots z_t e v_5 u_1 a.\]	

\item {\bf{$\gg^k$ is a strongly chordal graph}}

After deleting all the universal vertices in~$Y$ from~$\gg^k$,~$v_1$,~$v_2$, $v_3$, $v_4$, $v_5$ are all simple vertices. After deleting $Y \cup \{v_1, v_2, v_3, v_4, v_5\}$ from~$\gg^k$, vertices of~$X$ are all simple vertices.
	As simple vertices and universal vertices of a graph never participate in any induced holes or induced $j$-{\textsc{sun}}s for all integers $j\geq 3$, it holds that~$\gg^k$ is a strongly chordal graph if $\gg^k - (X\cup Y \cup \{v_1, v_2, v_3, v_4, v_5\})=\gg-\{v_1, v_2, v_3, v_4, v_5\}$ is, which is the case as shown in the proof for the case where $k = 0$.
\end{itemize}
	
Now we prove that there is a cycle~$C$ in~$\gg^k$ which is not $[t]$-extendable by contradiction. In particular, let
	\[C = a v_1 u_2 x_1 y_1 x_2 y_2 \cdots x_k y_k v_2 g u_1 v_5 e f v_4 u_3 v_3 b a.\]
	It is clear that~$C$ contains all vertices of~$\gg^k$ except $z_0, z_2, \dots, z_t$.
	Suppose for contradiction that~$C$ is $[t]$-extendable in~$\gg^k$.
	Note that $|Z| = t+1$. Then there exists a non-Hamiltonian cycle~$C'$ of~$\gg^k$ such that $V(C) \subset V(C')$ and, more precisely,~$V(C')$ is composed by $V(C)$ and a nonempty proper subset of~$Z$.  Let $J = \{u_2, u_3\}\cup Y$. Clearly, $J \subset V(C')$.
	The graph $\gg^k - J$ has exactly $k+1$ connected components, i.e., $\gg - \{u_2, u_3\}$ and the~$k$ isolated vertices in~$X$. We refer to Figure~\ref{fig:bbb} for an illustration of $\gg - \{u_1, u_2\}$. Then, as~$v_1$,~$v_2$,~$v_3$, and~$v_4$ are degree-$1$ vertices in $\gg - \{u_2, u_3\}$, and $|J| = k+2$, we know that  removing~$J$ from~$C'$ yields exactly $k+2$ paths, where~$k$ of them one-to-one correspond to the vertices of~$X$, and the other two paths, denoted~$P_1$ and~$P_2$, are vertex-disjoint and each contains exactly two of~$v_1$,~$v_2$,~$v_3$, and~$v_4$ as the ends. By the above discussion, $V(P_1)\cup V(P_2)$ consist of vertices in $\gg - Z - \{u_1,u_2\}$ and a nonempty proper subset of~$Z$.
	As~$v_5$ has degree two in $\gg - \{u_2, u_3\}$, with~$u_1$ and~$e$ being its two neighbors, we know that~$u_1$,~$v_5$, and~$e$ appear consecutively in one of~$P_1$ and~$P_2$. Due to symmetry, we need only to consider the following three cases. We shall show that all the three cases lead to some contradictions. Bear in mind that in the paths~$P_1$ or~$P_2$, the neighbor of~$v_1$ is~$a$, of~$v_2$ is~$g$, of~$v_3$ is~$b$, and of~$v_4$ is~$f$.
\begin{description}
\item[Case~1: $P_1$ is a $v_1$-$v_2$ path and $P_2$ is a $v_3$-$v_4$ path] \hfill

Due to the above discussion, in this case, the second and third vertices of~$P_1$, starting from~$v_1$, must be~$a$ and~$u_1$, respectively. Additionally, the second-last and the third-last vertices of~$P_1$ must be~$g$ and~$f$, respectively. However, this contradicts that~$f$ is the neighbor of~$v_4$ in~$P_2$, and~$P_1$ and~$P_2$ are vertex-disjoint.

\item[Case~2: $P_1$ is a $v_1$-$v_3$ path and $P_2$ is a $v_2$-$v_4$ path] \hfill

In this case, the third vertex of~$P_1$, starting from~$v_1$, is either~$b$ or~$u_1$.

We consider first the former case. In this case, $P_1=v_1 a b v_3$. Clearly,~$P_2$ cannot be $v_2 g f v_4$, since otherwise $Z\cap (V(P_1)\cup V(P_2))=\emptyset$, contradicting that $V(P_1)\cup V(P_2)$ contains a proper subset of~$Z$. Moreover, we know that~$u_1$ is in~$P_2$. Then, by the above discussion, the first five vertices of~$P_2$ are respectively~$v_2$,~$g$,~$u_1$,~$v_5$, and~$e$. The next vertex in~$P_2$ can be either~$f$ or~$z_t$. If it is~$f$, then $P_2=v_2 g u_1 v_5 e f v_4$. However, in this case $Z\cap (V(P_1)\cup V(P_2))=\emptyset$, a contradiction. So, it remains only the possibility that the sixth vertex of~$P_2$ is~$z_t$. However, it is easy to see that there is no path from~$z_t$ to~$v_4$ in the graph $\gg-(\{u_2, u_3, v_2, g, u_1, v_5, e\}\cup V(P_1))$, contradicting that~$P_2$, containing~$v_2$,~$g$,~$u_1$,~$v_5$, and~$e$, is a $v_2$-$v_4$ path that is vertex-disjoint from~$P_1$ in $\gg-\{u_2, u_3\}$.

Let us consider the latter case now. In light of the above discussions, the first five vertices of~$P_1$ must be $v_1$, $a$, $u_1$, $v_5$ and~$e$. The next vertex is either~$f$ or~$z_t$. It cannot be~$f$ because~$f$ is the neighbor of~$v_4$ in~$P_2$. However, it cannot be~$z_t$ either:  as $z_t z_{t-1} \cdots z_0 b v_3$ is the only $z_t$-$v_3$ path in $\gg-\{u_2, u_3, v_1, a, u_1, v_5, e\}$, it holds that $Z\subseteq V(P_1)$, which contradicts that $V(P_1)\cup V(P_2)$ contains a proper subset of~$Z$.

\item[Case~3: $P_1$ is a $v_1$-$v_4$ path and $P_2$ is a $v_2$-$v_3$ path] \hfill

As~$b$ is the neighbor of~$v_3$ in~$P_2$, we know that the third vertex in~$P_1$, starting from~$v_1$, must be~$u_1$. Recall also that~$f$ is the neighbor of~$v_4$ in~$P_1$. However, as $\{u_1, f\}$ separate~$v_2$ and~$v_3$,~$P_1$ and~$P_2$ cannot be two vertex-disjoint paths in $\gg-\{u_2, u_3\}$, a contradiction.
\end{description}
This completes the proof that~$C$ is not $[t]$-cycle extendable, and hence~$\gg^k$ is not $\ss$-cycle extendable.
	
	A Hamiltonian strongly chordal graph of~$n> 14+t+2k$ vertices which is not $\ss$-cycle extendable can be obtained from the graph~$\gg^k$ by adding $n-14-t-2k$ true-twins of~$v_1$.
		\end{proof}

\section{Affirmative Results}\label{sec-positive}
		
		This section is devoted to the proofs of Theorems~\ref{thm:fan-k5} and~\ref{thm:fan-free-power}.
		The following lemmas are from \cite{LafondS15}.
		
		\begin{lemma}\label{prop:common-neighbor} \cite{LafondS15}
			Let $C$ be a cycle of a chordal graph and let~$uv$ be an edge in~$C$. Then~$u$ and~$v$ have a common neighbor in~$V(C)$.
		\end{lemma}
		
		\begin{lemma}\label{prop:simplicial-vertex} \cite{LafondS15}
			Let~$G$ be a connected chordal graph, and let $S$ be a clique of~$V(G)$. If $G - S$ is disconnected, then each connected component of $G - S$ contains a simplicial vertex of~$G$.
		\end{lemma}
		
		It is well-known that if a chordal graph is not a clique, then it contains two non-adjacent simplicial vertices~\cite{Dirac1961}.
		Hence, Lemma~\ref{prop:simplicial-vertex} is extendable  to the case where $G - S$ is connected.
		
		\begin{corollary}\label{cor:simplicial-vertex}
			Let~$G$ be a connected chordal graph, and let~$S$ be a clique of~$V(G)$. Then each connected component of $G - S$ contains a simplicial vertex of~$G$.
		\end{corollary}

		\begin{lemma}\label{lemma:chordal-cycle}
			Let~$n$ be an integer such that all Hamiltonian chordal graphs of at most $n-1$ vertices are cycle extendable, and there exists a Hamiltonian chordal graph~$G$ of~$n$ vertices which is not cycle extendable.
			Let~$C$ be a nonextendable cycle of~$G$,~$Q$ be a connected component of $G-C$, and $S = N(Q)$ be the set of neighbors of~$Q$. Then the following conditions hold:
			\begin{enumerate}
				\item $S$ is not a clique of~$G$,
				\item any two vertices of~$S$ are not adjacent in~$C$, and
				\item if two vertices of~$S$ are adjacent in~$G$, then they have a common neighbor in~$Q$.
			\end{enumerate}
		\end{lemma}
		
		\begin{proof}
			The integer~$n$ stipulated in the lemma must exist, because any Hamiltonian chordal graph of at most four vertices are cycle extendable and there are nonextendable Hamiltonian chordal graphs of~$15$ vertices~\cite{LafondS15}.
			We prove the three statements of the lemma by contradiction.
			\begin{enumerate}
				\item For contradiction, assume that~$S$ is a clique of~$G$. Then, by Corollary~\ref{cor:simplicial-vertex},~$Q$ has a simplicial vertex, say~$v$. Then $G - v$ is Hamiltonian since~$G$ is Hamiltonian and all neighbors of~$v$ are pairwise adjacent in~$G$. Since~$C$ is not extendable in~$G$, and~$G$ is Hamiltonian, we know that~$C$ is not a Hamiltonian cycle of $G-v$. Then, as all Hamiltonian chordal graphs with vertices less than~$n$ are cycle extendable,~$C$ is extendable in $G - v$. This implies that~$C$ is also extendable in~$G$, a contradiction.
				
				\item For the sake of contradiction, assume that there are two vertices $x, y \in S$ such that $xy \in E(C)$. As $S = N(Q)$, both~$x$ and~$y$ have neighbors in~$Q$. However,~$x$ and~$y$ cannot have a common neighbor~$v$ in~$Q$: if this was the case, adding~$v$ to~$C$ and replacing the edge~$xy$ by the two edges~$xv$ and~$vy$ in~$C$ yield an extension of~$C$, contradicting that~$C$ is not extendable in~$G$.
				Now, there exist two distinct vertices~$x'$ and~$y'$ such that (1) $x'$ is a neighbor of~$x$ in~$Q$; (2)~$y'$ is a neighbor of~$y$ in~$Q$; and (3) $x'$ and~$y'$ have the shortest distance in~$G[Q]$ among all two distinct vertices fulfilling the first two conditions. However, every shortest $x'$-$y'$ path in~$G[Q]$ plus the three edges~$xx'$,~$yy'$, and~$xy$ yields a hole, a contradiction.
				
				\item Let $x,y\in S$ be two distinct vertices such that they are adjacent in~$G$. By Statement~2, it holds that $xy \notin E(C)$. Assume for contradiction that~$x$ and~$y$ do not have any common neighbors in~$Q$. Then, analogous to the above proof for Statement~2, there are two distinct vertices~$x'$ and~$y'$ in~$Q$ satisfying the same three conditions given above. However, any shortest $x'$-$y'$ path in~$G[Q]$ plus the three edges~$xx'$,~$yy'$, and~$xy$ yields a hole, a contradiction.
			\end{enumerate}
		\end{proof}

		Now, we are ready to prove Theorems~\ref{thm:fan-k5} and \ref{thm:fan-free-power}.
		
		\bigskip

		\begin{proof}[Proof of Theorem~\ref{thm:fan-k5}]
			Suppose for contradiction that there exist cycle nonextendable Hamiltonian $4$-{\textsc{fan}}-free graphs where every induced $K_5-e$ has true twins. Let~$G$ be such a graph with the minimum number of vertices. Clearly,~$G$ contains at least five vertices. Let~$C$ be a non-Hamiltonian cycle in~$G$ which is not extendable, let~$Q$ be a connected component of $G-C$, and let $S = N(Q)$. Clearly, $S \subseteq V(C)$.
			By Lemma~\ref{lemma:chordal-cycle}~(1),~$S$ is not a clique and hence $|S|\geq 2$. Let $x, y\in S$ be two non-adjacent vertices in~$S$ with the shortest distance on~$C$ among all pairs of non-adjacent vertices in~$S$.
			Let~$P$ be a shortest $x$-$y$ path on~$C$. There must be at least one inner vertex of~$P$ that is contained in~$S$, since otherwise a shortest $x$-$y$ path in $G[Q \cup \{x,y\}]$ plus a shortest $x$-$y$ path in~$G[V(P)]$ yields a hole.
			We break the discussion into two cases.
			
			\textbf{Case 1:} there is exactly one inner vertex of~$P$ contained in~$S$, say~$z$.
			
			By the selection of~$x$ and~$y$, we have that $xz, yz \in E(G)$. By Lemma~\ref{lemma:chordal-cycle}~(2), $xz, yz \notin E(C)$.	Let~$P_1$ be the path between~$x$ and~$z$ on~$P$, and~$P_2$ the path between~$y$ and~$z$ on~$P$. By Lemma~\ref{prop:common-neighbor}, there exists an inner vertex~$x'$ of~$P_1$ adjacent to~$x$ and~$z$, and an inner vertex~$y'$ of~$P_2$ adjacent to~$y$ and~$z$.
			Since~$z$ is the only inner vertex of~$P$ in~$S$, it holds that $x', y' \notin S$, i.e.,~$x'$ and~$y'$ are not adjacent to any vertex of~$Q$.
			Moreover, it holds that $xy', x'y, x'y' \notin E(G)$, since otherwise the shortest $x$-$y$ path in $G[\{ x,x',y,y' \}]$ plus a shortest $x$-$y$ path in $G[Q \cup \{x,y\}]$ is a hole.		

If~$x$,~$y$, and~$z$ have a common neighbor~$w$ in~$Q$, then $\{x',x,w,y,y',z\}$ induces a $4$\textsc{-fan}, in which~$z$ has degree~$5$, a contradiction.

Now we assume that $x$,~$y$, and~$z$ do not have any common neighbors in~$Q$. Then, by Lemma~\ref{lemma:chordal-cycle}~(3), we can find two distinct vertices~$u$ and~$v$ in~$Q$ such that (1)~$u$ is adjacent to~$x$ and~$z$ but not to~$y$, (2)~$v$ is adjacent to~$y$ and~$z$ but not to~$x$, and (3)~$u$ and~$v$ are two vertices with the shortest distance in~$G[Q]$ among all pairs of vertices in~$Q$ satisfying the first two conditions. Let~$P'$ be a shortest path between~$u$ and~$v$ in~$G[Q]$. Since $uz, vz \in E(G)$, to avoid a hole, all inner vertices of~$P'$ (if exist) are adjacent to~$z$. Moreover, by the selection of~$u$ and~$v$, all inner vertices of~$P'$ (if exist) are adjacent to neither~$x$ nor~$y$. Therefore, $x' x P' y$ is an induced path of length at least~$4$ in~$G$. Then, one can check that there are induced $4$-{\textsc{fan}}s in~$G$ (e.g., a subgraph induced by~$z$ and any five consecutive vertices in the path $x' x P' y$).

			\textbf{Case 2:} there are more than one inner vertex of~$P$ in~$S$.
			
			Let~$z$ and~$z'$ be two inner vertices of~$P$ in~$S$ such that~$z$ and~$z'$ are the nearest to~$x$ and~$y$ on~$P$ respectively, in terms of the length of the shortest paths between them. Clearly,~$z$ and~$z'$ are distinct. By the selection of~$x$ and~$y$, we have that $xz, xz', yz, yz', zz' \in E(G)$. By Lemma~\ref{lemma:chordal-cycle}~(2), $xz, zz', yz' \notin E(C)$.
			By Lemma~\ref{prop:common-neighbor}, there exists a vertex~$x'$ adjacent to~$x$ and~$z$, located between~$x$ and~$z$ on~$P$.
			Symmetrically, there exists a vertex~$y'$ adjacent to~$y$ and~$z'$, located between~$y$ and~$z'$ on~$P$. By the selection of~$z$ and~$z'$, we have that~$x'$ and~$y'$ are not adjacent to any vertex of~$Q$, i.e., $x',y' \notin S$.
			Moreover, it must be that $xy', x'y, x'y' \notin E(G)$, since otherwise there will be an $x$-$y$ path whose inner vertices are not adjacent to any vertex of~$Q$ and, moreover, this path plus a shortest $x$-$y$ path of $G[Q \cup \{x,y\}]$ is a hole.
			
			We claim that~$x$,~$y$, and~$z$ do not have any common neighbors in~$Q$. For contradiction, suppose they have a common neighbor~$w$ from~$Q$. Then $wz' \in E(G)$; otherwise,  $\{x,w,y,z'\}$ induces a hole. Now, $\{ z,z',x,w,y \}$ induces a $K_5 - e$, with~$xy$ being the missing edge. Then, there exists a true-twins pair of~$G$ from $\{z, z', x, w, y\}$. Notice that any true-twins pair in~$G$ is also a true-twins pair in any induced subgraph containing the true-twins pair. This restricts our focus only to the pairs $\{z, w\}$, $\{z', w\}$, and $\{z, z'\}$. The vertices~$z$ and~$w$ cannot be true twins, since otherwise the vertex next to~$z$ in~$P$ is also adjacent to~$w$, 
			contradicting with Lemma~\ref{lemma:chordal-cycle}~(2). Symmetrically,~$z'$ and~$w$ cannot be true twins either.
			If~$z$ and~$z'$ are true twins, then $zy' \in E(G)$ as $z'y' \in E(G)$. However, we arrive at a contradiction that $\{x', x, w, y, y', z\}$ induces a $4$\textsc{-fan}, in which~$z$ is the degree~$5$ vertex (notice that the selection of~$z$ and~$z'$ implies that~$x', y'\not\in S$, and hence~$w$ is adjacent to neither~$x'$ nor~$y'$).
			Therefore,~$z$ and~$z'$ are not true twins, and the proof for the claim that~$x$,~$y$, and~$z$ have no common neighbors in~$Q$ is completed.

However, if $x$,~$y$, and~$z$ do not have any common neighbors in~$Q$, we can find two vertices~$v$ and~$u$ exactly the same way as we did in the proof for Case~1, and arrive at a contradiction that the graph~$G$ contains a $4$-{\textsc{fan}} as an induced subgraph. 			
		\end{proof}
		
		\bigskip
		
			\begin{proof}[Proof of Theorem~\ref{thm:fan-free-power}]
			The proof is analogous to the proof of Theorem~\ref{thm:fan-k5} provided above. First, assume for the sake of contradiction that there is a Hamiltonian $\{4\text{-}{\textsc{fan}}, \overline{A}\}$-free chordal graph which is not cycle extendable. Let~$G$ be such a graph with the minimum number of vertices. Then, let~$C$,~$x$,~$y$,~$P$, and~$S$ be defined the same as in the proof of Theorem~\ref{thm:fan-k5}. Moreover, same to the proof of Theorem~\ref{thm:fan-k5}, we know that there must be at least one inner vertex in~$P$ that is contained in~$S$. Our proof proceeds by distinguishing the following two cases.
			
			\textbf{Case 1:} there is exactly one inner vertex of~$P$ in~$S$, say~$z$.

			The proof for this case is exactly the same as the one for Case~1 in the proof of Theorem~\ref{thm:fan-k5}. (The correctness of the proof relies only on the fact that the graph in consideration is a $4$-{\textsc{fan}}-free chordal graph.)
			
			\textbf{Case 2:} there are more than one inner vertex of~$P$ in~$S$.
			
Let~$z$, $z'$, $x'$, and $y'$ be defined the same as in Case~2 in the proof of Theorem~\ref{thm:fan-k5}.		
We first claim that~$x$,~$y$, and~$z$ do not have any common neighbors in~$Q$. For contradiction, suppose they have a common neighbor $w \in Q$. Same to the proof of Theorem~\ref{thm:fan-k5}, we can show that $wz' \in E(G)$ and $xy', yx', x'y'\not\in E(G)$. We have the following subcases to consider.
			\begin{itemize}
				\item  Case~2.1: If $x'z' \in E(G)$, then $\{x', x, w, y, y', z'\}$ induces a $4$\textsc{-fan}, in which~$z'$ is the degree-$5$ vertex.
				\item  Case~2.2: If $y'z \in E(G)$, then $\{x', x, w, y, y', z\}$ induces a $4$\textsc{-fan}, in which $z$ is the degree-$5$ vertex.
				\item  Case~2.3: If $x'z',y'z \notin E(G)$, then $\{x', x, z, z', y, y'\}$ induces a~$\overline{A}$ in which~$z$ and~$z'$ are the degree-$4$ vertices.
			\end{itemize}
			Therefore,~$x$,~$y$, and~$z$ do not have any common neighbors in~$Q$.
	
If $x$,~$y$, and~$z$ do not have any common neighbors in~$Q$, we can find two vertices~$v$ and~$u$ exactly the same way as we did in Case~1 of the proof of Theorem~\ref{thm:fan-k5}, and arrive at a contradiction that the graph~$G$ contains a $4$-{\textsc{fan}} as an induced subgraph.		
		\end{proof}

\section{Concluding Remarks}\label{sec-conclusion}
In 1990, Hendry conjectured that every Hamiltonian chordal graph is cycle extendable~\cite{Hendry90}. In 2015, this conjecture was refuted by Lafond and Seamone~\cite{LafondS15} who, at the same time, proposed two new questions (Questions~1 and~2). Later, Arangno~\cite{Arangnothesis} proposed a stronger conjecture. In this paper, we refuted all these conjectures and questions by providing many counterexamples that even satisfy further conditions (Theorems~\ref{thm:counter-k-connected}). To complement these negative results, we confirmed Hendry's conjecture for Hamiltonian $4$-{\textsc{fan}}-free chordal graphs where every induced $K_5-e$ has true twins (Theorem~\ref{thm:fan-k5}), and for Hamiltonian $\{4\textsc{-fan}, \overline{A}\}$-free chordal graphs (Theorem~\ref{thm:fan-free-power}).

Our study arises two questions for future research. First, given that Hamiltonian $3$-{\textsc{fan}}-free chordal graphs are cycle extendable~\cite{Gerekthesis}, it is interesting to study whether Hendry's conjecture holds for Hamiltonian $4$-{\sc{fan}}-free chordal graphs. Second, as discussed earlier, our results imply that Hendry's conjecture holds for Hamiltonian $4$-leaf powers (Corollary~\ref{cor:4-leaf-power}), extending the results of Gerek~\cite{Gerekthesis} who showed that Hendry's conjecture holds for Hamiltonian $k$-leaf powers for all $k\in [3]$. So, another interesting question is whether Hendry's conjecture holds for Hamiltonian $5$-leaf powers. We reiterate that in general $i$-leaf powers are $(i+2)$-leaf powers, and for each $i\in [3]$, $i$-leaf powers are $(i+1)$-leaf powers. However, there are $4$-leaf powers which are not $5$-leaf powers~\cite{DBLP:journals/ipl/BrandstadtL06,DBLP:conf/isaac/FellowsMRST08}.

\bigskip

\section*{Acknowledgement} We would like to thank an anonymous referee for many useful suggestions that have improved the presentation of this paper.


\end{document}